\documentclass[12pt,a4paper]{article}
\usepackage{graphicx} 
\usepackage{subcaption} 
\usepackage{physics}
\usepackage[backend=bibtex,style=alphabetic,doi=false,isbn=false,url=false,maxbibnames=20,sorting=none]{biblatex}
\usepackage{amsmath}
\usepackage{mathtools}
\usepackage{amssymb}
\usepackage{physics}
\usepackage{ dsfont }
\usepackage{braket}
\usepackage{siunitx}
\usepackage{graphicx}
\usepackage{enumerate}
\usepackage{authblk}
\usepackage{sidecap}
\usepackage{float}
\usepackage[toc,page,title]{appendix}
\usepackage{multirow}
\usepackage{amsthm}
\usepackage{sidecap}
\usepackage{eucal}
\usepackage{comment}
\usepackage{enumerate}
\usepackage{setspace}
\usepackage[super]{nth}
\usepackage{hyperref}
\usepackage{todonotes}
\usepackage{thmtools}
\usepackage{thm-restate}
\usepackage{amsmath}
\usepackage{mathtools}
\usepackage{amsfonts}
\usepackage{amssymb}
\usepackage{braket}
\usepackage{siunitx}
\usepackage{graphicx}
\usepackage{enumerate}
\usepackage{float}
\usepackage[toc,page,title]{appendix}
\usepackage{multirow}
\usepackage{amsthm}
\usepackage{eucal}
\usepackage{comment}
\usepackage{setspace}
\usepackage{todonotes}
\usepackage{dirtytalk}
\usepackage{authblk} 
\usepackage{tikz}
\usepackage{tcolorbox} 
\newtcolorbox{mybox}[1][]{
boxrule=0.5pt, colframe=black, sharp corners, #1}
\newtcolorbox{maitebox}[1][]{colback=cyan!30, boxrule=0.5pt, colframe=black, sharp corners, #1}

\usepackage[capitalise]{cleveref}
\Crefname{lemma}{Lemma}{Lemmas}
\Crefname{proposition}{Proposition}{Propositions}
\Crefname{definition}{Definition}{Definitions}
\Crefname{theorem}{Theorem}{Theorems}
\Crefname{conjecture}{Conjecture}{Conjectures}
\Crefname{corollary}{Corollary}{Corollaries}
\Crefname{example}{Example}{Examples}
\Crefname{section}{Section}{Sections}
\Crefname{appendix}{Appendix}{Appendices}
\Crefname{figure}{Fig.}{Figs.}
\Crefname{equation}{Eq.}{Eqs.}
\Crefname{table}{Table}{Tables}
\Crefname{item}{Property}{Properties}
\Crefname{remark}{Remark}{Remarks}
\Crefname{axioms}{Axioms}{Axioms}

\crefformat{footnote}{#2\footnotemark[#1]#3}

\newcommand{\newreptheorem}[2]{%
  \newenvironment{rep#1}[1]{%
    \def\rep@title{#2 \ref{##1}}%
    \begin{rep@theorem}%
  }{\end{rep@theorem}}%
}

\newtheorem{thm}{Theorem}
\newreptheorem{thm}{Theorem}

\newtheorem{lemma}[thm]{Lemma}
\newreptheorem{lemma}{Lemma}

\newreptheorem{col}{Corollary}

\crefformat{footnote}{#2\footnotemark[#1]#3}

\newtheorem*{rep@theorem}{rep@title}

\definecolor{ma}{rgb}{246,76,246}
\definecolor{ha}{RGB}{246,76,246}
\definecolor{ma}{rgb}{0.858, 0.188, 0.478}

\newcommand{\bet}{f}
\newcommand{\RQ}{F}
\newcommand{\bone}{\mathds{1}}

\bibliography{bibliography}
\title{A resource theory of gambling}

\author[1,2]{Maite Arcos\thanks{\texttt{maite.arcose@gmail.com}}}
\author[3]{Renato Renner}
\author[1]{Jonathan Oppenheim}

\affil[1]{Department of Physics and Astronomy, University College London, UK}
\affil[2]{Institute of Physics, Ecole Polytechnique Fédérale de Lausanne (EPFL), Lausanne, Switzerland}
\affil[3]{Institute for Theoretical Physics, ETH Zurich, Zurich, Switzerland}

\date{September 12, 2025}

\begin{document}
\maketitle

\begin{abstract}
Betting games provide a natural setting to capture how information yields strategic advantage. 
The Kelly criterion for betting —long a cornerstone of portfolio theory and information theory—admits an interpretation in the limit of infinitely many repeated bets. We extend Kelly’s seminal result into the single‐shot and finite‐betting regimes, recasting it as a resource theory of adversarial information. This allows
one to quantify what it means for the gambler to have more information than 
the odds-maker. Given a target rate of return, after a finite number of bets, we compute the optimal strategy which maximises the probability of successfully reaching the target, revealing a risk-reward trade-off characterised by a hierarchy of Rényi divergences between the true distribution and the odds. The optimal strategies
in the one-shot regime  coincide with strategies maximizing expected utility, and minimising hypothesis testing errors, thereby bridging economic and information-theoretic viewpoints. We then generalize this framework to a distributed side‐information game, in which multiple players observe correlated signals about an unknown state. Recasting gambling as an adversarial resource theory provides a unifying lens that connects economic and information-theoretic perspectives, and allows for generalisation to the quantum domain, where quantum side-information and entanglement play analogous roles.

\end{abstract}


\tableofcontents

\pagebreak

\section{Introduction}

Adversarial interactions are fundamental to information processing—from strategic games and financial markets, to cryptographic security and quantum communication. In such settings, agents with conflicting goals compete by exploiting asymmetries in information. Quantifying the operational advantage that one agent gains over another thus lies at the heart of both strategic reasoning and information theory.

Betting provides an operational framework for quantifying informational advantage. In a gambling game, alignment between a gambler’s bets and the realised outcomes determines their measurable advantage—the closer the bets track the empirical outcomes, the greater the gain. This direct link between information and payoff makes gambling a primitive setting in which abstract notions of knowledge, uncertainty, and prediction acquire concrete, operational meaning.

The connection between information and gambling was first formalised by Kelly~\cite{kellyspaper}, who showed that in the asymptotic limit of many repeated bets, the optimal growth rate of a gambler’s wealth is governed by the relative entropy between the true outcome distribution and the odds set by the bookmaker. Kelly’s result established a bridge between Shannon’s information theory and rational decision-making under uncertainty, revealing that information gain directly translates into observable payoffs.

Despite its conceptual elegance, the Kelly framework also exposed a long-standing divide between economic and information-theoretic perspectives. While Kelly’s criterion is celebrated for its asymptotic optimality, economists—most notably Samuelson~\cite{samuelson1971fallacy, samuelson1979why}—criticised it for being a special case (maximising expected log-wealth) that fails to account for the spectrum of individual risk preferences as formalised in expected utility theory~\cite{vonneumann1944gametheory}. This tension between information-based optimality and utility-based rationality has persisted for decades, highlighting the absence of a unified operational account.

In this work, we resolve this tension by extending Kelly’s framework to the finite and single-shot regimes using tools from information theory~\cite{csiszar1998method}. This extension reveals a fundamental risk–reward trade-off governed by a hierarchy of Rényi divergences, with an agent’s degree of risk aversion selecting the optimal divergence order. We show that these strategies are mathematically equivalent to those that maximise expected utility for agents with constant relative risk aversion, thereby unifying the information-theoretic and economic perspectives within a single operational formalism. We further generalise this setting to distributed adversarial games with correlated side information, where Nash equilibria emerge asymptotically; demonstrating a connection to Harsanyi's theory of games with incomplete information \cite{harsanyi1967games1}. 

By casting gambling in the language of resource theories \cite{Oppenheim2010Quantumness, brandao_resource_2015, gour_chitambar_art_2018,gour2025quantum}, we provide a unifying framework where informational advantage becomes a tangible resource that can be quantified, transformed, and compared. This approach naturally connects economic decision theory (via Blackwell's theorem \cite{blackwell1953equivalent}) with physical resource theories, revealing they share the same underlying principle: an agent's operational power is determined by what transformations their informational resources permit. To further illustrate this point, in concurrent work \cite{AdversarialThermo}, we show that the results derived here also find application in thermodynamics, demonstrating that the structure found in gambling also underlies thermodynamic work extraction. Crucially, we also establish the formal foundation for a quantum generalisation. 

The paper is organised as follows. In Section \ref{modernkelly}, we review Kelly’s original framework and extend it to the single-shot and finite-horizon regimes, highlighting the resulting trade-offs between growth and risk. Given a target growth rate, we derive the betting strategy which optimises for achieving this rate. This gives a risk-reward trade-off. An important element of this is to reframe gambling in terms of the empirical realised type class of outcomes that actually occur.
We then connect these results to expected utility theory in Section \ref{equivalencewithclassicalswisspeople}, showing how different degrees of risk aversion correspond to distinct betting strategies, thereby linking the analysis to Rényi divergences and asymmetric hypothesis testing. In Section \ref{doublesideinfo}, we generalise the setting to distributed adversarial scenarios with side information, proving that Nash equilibria emerge asymptotically without invoking rationality axioms. This will allow us to generalise to the quantum scenarios \cite{ARTArcos2025adversarial}. The final casting of classical gambling as an adversarial resource theory is contained in Section \ref{sec:GamblingRT}. Finally, we discuss implications and open problems in the Conclusion. Technical details and background material are deferred to the Appendices.

\subsection{Notation and terminology}\label{gamblingnotation}

Let $X$ be a random variable taking values in an alphabet $\mathcal{X}$. In keeping with tradition in the gambling literature we will sometimes refer to $X$ as the outcome of a horse race. The probability distribution for $X$ is denoted as $P_{X}=\{ p_{X}(x) : x \in \mathcal{X} \}$, with $p_{X}(x)= Pr \{ X=x\}$. When there is no ambiguity, $p_{X}(x)$ will sometimes be abbreviated as $p(x)$. The support of $X$, denoted by $\mathcal{S}_{X}$, is the set of all $x \in \mathcal{X}$ such that $p(x)>0$. Given any function $g(X)$ of the random variable $X$, we will denote its expected value as $E(g(X))= \sum_{x} p(x)g(x)$. 

Let $x^{n}$ be an n-tuple of elements from the alphabet $\mathcal{X}$. The frequency distribution $\lambda_{x^{n}}$ of $x^{n}$ is the probability distribution on $\mathcal{X}$ defined by the relative number of occurrences of each symbol. We will refer to the frequency distribution $\lambda_{x^{n}}$ as the \textit{type} of the sequence $x^{n}$. We will denote the set of all types by $\mathcal{Q}_{n}$. For any type $ \lambda_{x^{n}} \in \mathcal{Q}_{n}$, we denote by $\Lambda_{n}^{\lambda_{x^{n}}}$ the corresponding \textit{type class}, i.e. the set of all $n$-tuples $ x^{n}= (x^{1},... ,x^{n})$ with frequency distribution $\lambda_{x^{n}}$. When we write expressions such as
$P_X(\Lambda_{x^n})$, this denotes the probability of the entire type class
$\Lambda_{n}^{\lambda_{x^{n}}}$. We denote by $\log$ the logarithm to the base two. 

For two probability distributions $P_{X}$, $Q_{X}$ on some set $\mathcal{X}$, the Rényi divergence of order $\alpha$ is defined for all $\alpha \in \mathbb{R}$ as 
\begin{equation}\label{RényiAlphaClass}
    D_{\alpha}(P_{X}||Q_{X})= \frac{\mathrm{sgn}(\alpha)}{\alpha -1} \log \bigg( \sum_{x} p_{X}(x)^{\alpha} q_{X}(x)^{1-\alpha}  \bigg) 
\end{equation}
We will refer to 
\begin{equation}
    D_{1}(P_{X}||Q_{X}) \equiv \lim_{\alpha^{+} \rightarrow 1} D_{\alpha}(P_{X}||Q_{X}) = \sum_{x} p_{X}(x) \log \bigg(\frac{p_{X}(x)}{q_{X}(x)} \bigg)
\end{equation}
 as the relative entropy (though it is also often called the Kullback-Leibler or KL divergence). For $\alpha= \infty$ and $\alpha=0$, we define $D_{\alpha}(P_{X}||Q_{X})$ analogously. We will denote the Shannon entropy of the probability distribution $P_{X}$ as $H(P_{X})$, where
 \begin{equation}
 H(P_{X})= -\sum_{x} p(x) \log(p(x)).
 \end{equation}



\section{Finite Kelly betting in terms of the a posteriori distribution}\label{modernkelly}

\subsection{Kelly's original framework}

The Kelly \cite{kellyspaper} betting scheme is a classical information-theoretic scenario in which the gambler, Alice, allocates fractions of her wealth to the possible outcomes of a horse race, always distributing a non-zero fraction to each outcome (in order to avoid ever being completely broke). The adversary, Bob, sets odds for each possible outcome. 

Suppose the horse race is described by a random variable $X$. Let $Q_X^A$ denote Alice’s betting distribution (i.e., $Q_X^A(x)$ is the fraction of wealth she allocates to outcome $x$), and let $Q_X^B$ denote Bob’s odds distribution. With this notation, if outcome $x$ occurs in a given round, Alice’s wealth is multiplied by a factor $\tfrac{Q_X^A(x)}{Q_X^B(x)}$. 

Assuming Alice reinvests her wealth after each round, her wealth after $n$ rounds is
\begin{equation}
    W_{n}= W_{i}\prod_{x}\left(\frac{Q_{X}^{A}(x)}{Q_{X}^{B}(x)}\right)^{N_{x}},
    \label{eq:reinvesting}
\end{equation}
where $N_x$ is the number of times outcome $x$ occurred in the $n$ rounds and $W_i$ is her initial wealth.

In the limit $n \to \infty$, the ratio of Alice’s final wealth to her initial wealth satisfies 
\begin{equation}\label{KellyCT}
   \boxed{ \frac{W_{F}}{W_{i}}
    \;=\; 2^{{\,n\!\left(D(P_{X}\,\|\,Q_{X}^{B}) \;-\; D(P_{X}\,\|\,Q_{X}^{A})\right)} }},
\end{equation}
where $P_X$ is the true distribution of the race outcomes.  $W_n/W_i$ is sometimes called the {\it wealth relative}. Since relative entropy is non-negative, Alice’s optimal strategy is to bet according to $Q_X^A=P_X$, ensuring that her wealth grows at the maximum possible rate. Equation~\eqref{KellyCT} is the classical result of Kelly \cite{kellyspaper}.

\subsection{Single-shot and finite $n$ Kelly betting }\label{finitekelly}

Before starting the discussion on finite-$n$ Kelly betting, it is worth clarifying that there is a discrepancy in the way that the equation 
\[
\frac{W_F}{W_i} = \exp\left(n \left( D(P_X \| Q_X^B) - D(P_X \| Q_X^A) \right) \right)
\]
is interpreted in information-theoretic and economic contexts. In information-theoretic contexts, this equation is often interpreted as indicating that winning the Kelly betting game reduces to the ability of a gambler and adversary to estimate the true probability distribution of the random variable \cite{singleshotkellydifferentderivation, Kull2022ANovelKellyInterpretation}. This interpretation aligns with how the relative entropy (or Kullback-Leibler divergence) is used in many information-theoretic tasks, where it quantifies the inefficiency of approximating one distribution with another.

However, from an expected utility point of view, and in the original Kelly scenario itself, the true probability distribution $P_X$ of the random variable (in this case, the horse races) is  known. The expected utility hypothesis, a cornerstone of economic decision theory, posits that individuals make choices to maximize their expected utility, and that this expected utility is a measure of satisfaction. This framework incorporates risk aversion by assigning lower utility to riskier outcomes. 

The aim of this section is to re-formulate the Kelly betting scenario for a finite number of horse races. Whilst the traditional Kelly scenario relies on the asymptotic convergence of the observed sequences of outcomes to the true distribution of the races, in the formulation below, we consider a gambler who infers strategies in terms of the possible empirical frequency distributions of observed outcomes (i.e. the possible types of the {\it a posteriori} distribution). For finite numbers of horse races, we show that the gambler faces a risk-reward tradeoff between wealth growth and probability of success, which converges to Kelly's \cite{kellyspaper} well-known result in the asymptotic limit. 


We now formulate the Kelly's scenario in the finite-size regime.  We consider the situation where Alice is interested in betting on the outcome of $n$ horse races, where $n$ is finite. Let $X$ be a random variable taking values in an alphabet $\mathcal{X}$, which describes the outcome of the horse races. Suppose that there are $k$ horses. Let $X_{1}, X_{2},...X_{n}$ be drawn i.i.d. according to the probability distribution $P_{X}$. 

A key insight is that Kelly’s repeated i.i.d.\ betting process is equivalent to a single bet on the entire sequence $x^n \equiv x_1 \ldots x_n$. This is evident from ~\eqref{eq:reinvesting}, which shows that after $n$ rounds, the ratio of Alice’s initial wealth to final wealth can be written as
\begin{equation}
    \frac{W_n}{W_i} 
    = \frac{\prod_x Q_X^A(x)^{N_x}}{\prod_x Q_X^B(x)^{N_x}}
    = \frac{Q_X^A(x^n)}{Q_X^B(x^n)},
\end{equation}
with Alice  allocating fractions of her money to each of the possible strings $x^{n} \equiv x_{1}...x_{n}$.
In other words, the repeated-betting process is equivalent to a 
\emph{single bet on the entire sequence} $x^n$, where the stake is
\[
Q_X^A(x^n) = \prod_x Q_X^A(x)^{N_x},
\]
and the odds are
\[
Q_X^B(x^n) = \prod_x Q_X^B(x)^{N_x}.
\]








Since all sequences with the same type $\lambda_{x^n}$ occur with the same probability, a rational strategy should allocate the same fraction of wealth to them. A natural way to enforce this is to choose a letter-wise betting distribution $Q_X^A(x)$ and allocate bets to strings according to $Q_X^A(x^{n})$. A well-known result from the method of types \cite{csiszar1998method} tells Alice that the fraction to allocate to each possible \textit{string} is given by: 
\begin{equation}\label{thetype}
   Q_{X}^{A}(x^{n})= 2^{-n(H(\lambda_{x^{n}}) + D(\lambda_{x^{n}}||Q_{X}^{A}))}
\end{equation}
%
Expression \eqref{thetype} tells us that the fraction allocated to each sequence depends only on its type, as does its probability. 
By the same argument, we see that if Bob allocates odds according to the distribution $Q^{B}_{X}$ for each outcome, he will assign 
\begin{equation}
   Q^{B}_{X}(x^{n})= 2^{-n(H(\lambda_{x^{n}}) + D(\lambda_{x^{n}}||Q^{B}_{X}))}
\end{equation}
to each sequence of races of type $\lambda_{x^{n}}$.

When the outcome of the i.i.d.\ source is revealed to be a string belonging to type class $\lambda_{x^{n}}$, Alice's wealth therefore grows according to the ratio
\begin{equation}\label{singleshotKelly}
   \frac{W_{F}}{W_{i}}=  \frac{Q_{X}^{A}(x^{n})}{Q_{X}^{B}(x^{n})}= 2^{n(D(\lambda_{x^{n}}||Q^{B}_{X})-D(\lambda_{x^{n}}||Q^{A}_{X}))}
\end{equation}
for any value of $n$. 
Since $\lambda_{x^{n}}$ converges to the true distribution as $n$ goes to infinity, we recover Kelly's result in this limit. 

In the finite-size regime, the ratio of Alice's final to initial wealth is a random variable. Equation \ref{singleshotKelly} shows, in particular, that when Alice's allocation matches the type of the observed sequence, her wealth grows according to 
\begin{equation}
   \frac{W_{F}}{W_{i}}=  2^{nD(\lambda_{x^{n}}||Q^{B}_{X})}
   \label{eq:WealthIfCorrect}
\end{equation}
Since different types occur with different probabilities, Alice's choice of strategy inevitably depends on her willingness to compromise between the amount of money she would make if she guessed the type of the sequence correctly and the probability of that sequence/type. This choice reflects Alice's risk aversion. For example, an extremely risk-averse Alice might choose the strategy $Q_{X}^{A}=Q_{X}^{B}$, ensuring that her wealth remains constant with probability one and therefore avoiding any risk.

Alice can use the same reasoning to calculate the probability of the sequence $x^{n}$: 
\begin{equation}
    P_{X}(x^{n}) = 2^{-n\!\left( H(\lambda_{x^{n}})+D(\lambda_{x^{n}}\|P_{X}) \right)}
\label{eq:ProbSeq}
\end{equation}

Since all sequences of the same type $\lambda_{x^n}$ have identical probability, 
the total probability that the outcome is a string in that type class is
\begin{equation}
    P_{X}(\lambda_{x^{n}}) 
    = |\Lambda_{\lambda_{x^n}}|\; P_{X}(x^{n}), \qquad x^n\in \Lambda_{\lambda_{x^n}}, \label{eq:TypeProbExact}
\end{equation}
where $|\Lambda_{\lambda_{x^n}}|$ is the size of the type class.
It is well known \cite{CoverThomas2006} that
\begin{equation}
    \frac{1}{(n+1)^{|\mathcal{X}|}}\, 2^{\,n H(\lambda_{x^n})}
    \;\leq\; |\Lambda_{\lambda_{x^n}}|
    \;\leq\; 2^{\,n H(\lambda_{x^n})}.
    \label{eq:TypeSizeBounds}
\end{equation}
Thus, up to subexponential factors in $n$, we may write  the standard large-deviation estimate
\begin{equation}
    P_{X}(\lambda_{x^{n}}) \doteq 2^{-n\,D(\lambda_{x^{n}}\|P_{X})}
    \label{eq:probtype}
\end{equation}
where $\doteq$ denotes equality up to sub-exponential terms.



The potential reward for type
$\lambda_{x^n}$ is given by $ \exp( n D(\lambda_{x^n}\|Q_X^B))$, while from
\cref{eq:probtype} its probability is $\exp(-nD(\lambda_{x^n}\|P_X))$. Each type $\lambda$ thus corresponds to a
risk–reward pair. Because Alice cannot know in advance which type will
occur, she must select a bet $Q_X^A$, and her choice
determines which types she aligns most closely with.

 A natural problem we may encounter is that after $n$ rounds, Alice may wish to achieve a return on investment of at least $R_n$, i.e. she would like 
\begin{align}
\frac{1}{n}\log \frac{W_F}{W_i}\geq R_n\,.
\end{align}
Given this constraint, we want to find the optimal betting strategy which optimises the probability $\epsilon(R_n)$ of achieving this rate of return. The curve $R_n$ vs $\epsilon$ is the risk reward trade-off. Equivalently, we can first fix the probability of success $\epsilon$, and then maximise her achievable reward, in the process finding the optimal betting strategy which achieves a return of at least $R_n(\epsilon)$. This is just the inverse $R_n(\epsilon(R_n))=R_n$. We will consider this equivalent formulation, because as we shall soon see, it  corresponds most naturally to hypothesis testing, a task with a deep connection to resource theories\cite{ogawa2002strong,hayashi2003general,nagaoka2007information,BrandaoPlenio2010,tomamichel2013hierarchy,li2014second,brandao_resource_2015,wang2019resource,buscemi2019information,sagawa2021asymptotic,gour_chitambar_art_2018,gour2025quantum}.

If Alice demands that the probability of success exceeds some $\epsilon$, then this corresponds to sets $A_\lambda$ of type classes, such that
\begin{equation}\label{constraintforprobability}
    \sum_{\lambda\in A} P_X(\lambda_{x^n})\geq \epsilon\,.
\end{equation}
Under this constraint, she then wishes to find the optimal betting strategy, and set $A_\lambda$  such that her return $R_n$ is as large as possible. In Appendix \ref{app:sanov} we show that up to sub-exponential terms, we can reformulate the constraint of Eq \eqref{constraintforprobability} as one which corresponds to finding single type classes which satisfy
\begin{equation}\label{eq:RiskConstraint}
    P_X(\lambda_{x^n})
    \;\doteq \; 2^{{-n\,D(\lambda_{x^n}\,\|\,P_X)}}
    \;\geq\; \epsilon \,.
\end{equation}
and then maximising her achievable reward which corresponds to $D(\lambda_{x^n}\|Q_X^B)$ under this constraint. 

Mathematically, these trade-offs can be expressed as constrained optimisation problems. 
For instance, certain individuals may wish to maximise the payoff 
$D(\lambda_{x^{n}} \,\|\, Q_{X}^{B})$ subject to a constraint on the probability of error
$D(\lambda_{x^{n}} \,\|\, P_{X})$, whilst others might reverse the roles. In either case, the optimisation is over the empirical 
distribution (type) $\lambda_{x^{n}}$.
These problems are well-understood in classical information theory: they amount to finding the most 
likely type $\lambda$ that balances two competing relative entropy terms. A standard 
Lagrange multiplier argument, which we describe in detail in \cref{lagrangemultipliers}, shows that the optimal bet is\footnote{See, e.g., \cite{CoverThomas2006} for an overview of this  in the context of hypothesis testing.}: 
\begin{equation}\label{OptimiserHT}
    Q_{X}^{A*}(x) \;=\; 
    \frac{P_{X}(x)^{\eta}\, Q_{X}^{B}(x)^{\,1-\eta}}
         {\sum_{x'} P_{X}(x')^{\eta}\, Q_{X}^{B}(x')^{\,1-\eta}}.
\end{equation}
Here $\eta$ is determined by the constraint \cref{eq:RiskConstraint}: it plays the role of a Lagrange 
multiplier that interpolates between $P_{X}$ and $Q_{X}^{B}$. Intuitively, 
$\eta = 1$ recovers $P_{X}$, $\eta = 0$ recovers $Q_{X}^{B}$, and intermediate values 
yield exponential mixtures of the two. Thus, the optimiser $Q_X^{A*}$ selects the effective 
distribution that generates the most wealth given the gambler’s risk preference.  

Recall that Alice succeeding with probability  at least $\epsilon$, is equivalent to the probability that the empirical type class $\lambda_{x^n}$ matches her bet $Q_X^{A*}$. We thus obtain the following result:
\begin{mybox}
For a bet which must succeed with probability at least $\epsilon$, the ratio of 
Alice's initial wealth to final wealth is at least
\begin{align}
    \frac{1}{n}\log\frac{W_F}{W_i}&\geq{D(Q_X^{A*}||Q_X^B)}\nonumber\\
    &\geq D_{\lambda(\epsilon)}(P_X||Q_X^B)+\frac{\lambda(\epsilon)}{1-\lambda(\epsilon)}\frac{\log{\epsilon}}{n}
\end{align}
where the Rényi-divergence is defined in Eqn \eqref{RényiAlphaClass}
and $\eta \in(0,1)$.
We have used the identity
\begin{equation}\label{eq:reward-identity}
D\!\left(Q_X^{A*}\,\|\,Q_X^B\right) 
=  D_\lambda\!\left(P_X \,\|\, Q_X^B\right)-\frac{\lambda}{1-\lambda}\, D\!\left(Q_X^{A*}\,\|\,P_X\right) 
\end{equation}
(see Appendix \ref{sec:thesimplex}). This not only defines a risk-reward trade-off, quantified by the pairs $\{ \epsilon,D_{\lambda(\epsilon)}\}$, it also gives an optimal strategy for achieving it. Namely, Alice should bet according to Eq.~\eqref{OptimiserHT}.
\end{mybox}
The identity, Eq. \eqref{eq:reward-identity} itself gives a risk-reward trade-off for gambling, since the left hand side quantifies the reward, and the first term on the right quantifies the risk.
For a given risk tolerance $\epsilon$, the constraint of Eq \eqref{eq:RiskConstraint} $D(\lambda_{x^n} || P_X) \leq \frac{1}{n}\log \epsilon$ defines a set of acceptable types $\lambda_{x^n}$. The optimal trade-off is found by solving the constrained optimization over this set, which introduces a Lagrange multiplier $\lambda$. This multiplier $\lambda$ is therefore a function of the chosen risk level, i.e., $\lambda = \lambda(\epsilon)$. The optimizer $Q_X^{A*}$ in Eq. \eqref{OptimiserHT} is subsequently determined by this $\lambda(\epsilon)$. Finally, this value of $\lambda$ defines the specific Rényi divergence $D_{\lambda(\epsilon)}(P_X || Q_X^B)$ that characterizes the achievable reward in the resulting risk-reward pair ${\epsilon, D_{\lambda(\epsilon)}(P_X || Q_X^B)}$. This chain of dependence ($\epsilon \rightarrow \lambda \rightarrow Q^{A*} \rightarrow D_\lambda$) precisely quantifies how a gambler's risk aversion shapes their optimal strategy and its resulting payoff profile.


As expected, this family includes the cases where Alice’s
strategy is proportional to the odds ($Q_X^A=Q_X^B$), or to the true
distribution ($Q_X^A=P_X$). In the following section, we relate this
solution and its parameter $\lambda$ to expected utility formulations
and degrees of risk aversion.

\subsection{Connection to Expected Utility Theory}\label{equivalencewithclassicalswisspeople}





The expected utility hypothesis, proposed by von Neumann and Morgenstern \cite{vonneumann1944gametheory} and a cornerstone of economic decision theory, states that rational individuals make decisions to maximize their expected utility rather than the expected value of payoffs. In this framework, a utility function captures the agent’s preferences over outcomes, so that two gambles with the same expected value may nevertheless be ranked differently. This makes it possible to formalise the idea that individuals may prefer safer or riskier options depending on their attitudes toward uncertainty.

In the classical presentation of expected utility theory, one compares preferences over externally specified lotteries. For example, a risk-averse person might prefer a guaranteed 100 over a 50 per cent chance of winning 250, even though the latter has a higher expected value. In the gambling framework considered here, the situation is more nuanced: the agent’s betting strategy itself defines the lottery. Choosing a distribution $Q_X^A$ against odds $Q_X^B$ induces the distribution of wealth outcomes, which means the decision problem is not merely to rank fixed lotteries but to optimise over a family of lotteries generated by strategic choices. Risk aversion is therefore directly intertwined with strategy design, as different utility functions correspond to different optimal allocations.

From this angle, the expected utility hypothesis is best understood as prescribing an optimisation principle: given a utility function encoding risk attitude, the agent selects the strategy $Q_X^A$ that maximises expected utility of wealth. Different utility functions induce different optimal strategies, interpolating between safe and aggressive allocations. This framing makes explicit how “utility” in our setting depends on strategy, rather than being merely a passive evaluation of uncertainty.

When initially published, Kelly’s \cite{kellyspaper} result faced criticism for two main reasons. First, it was interpreted by some as claiming that the strategy of proportional gambling (that which maximises the expected logarithm) is optimal for all individuals, regardless of their risk preferences. This interpretation overlooks the fact that different gamblers may have varying levels of risk tolerance as proposed in expected utility theory, leading to suboptimal outcomes for risk-averse individuals. Second, Kelly’s framework lacks validity in the short-term (finite-size) regime, where the gambler faces a limited number of bets and cannot rely on the law of large numbers to average out fluctuations. 

Nobel laureate Paul Samuelson was among the most prominent critics of Kelly’s proposal. In a series of papers \cite{samuelson1971fallacy, samuelson1979why}, he argued that the Kelly criterion could not be universally valid, since no single strategy can simultaneously maximise all possible utility functions. Samuelson’s aim was to demonstrate, in a mathematically rigorous way, that proportional gambling is not a general principle of rational choice. This line of critique was less about the internal consistency of Kelly’s result—log-utility maximisation is uncontroversial—and more about resisting claims that Kelly betting should apply across the board to all investors. The debate has continued \cite{CarrCherubini2022, Poundstone2005, Ziemba2012Response}, with researchers seeking broader formulations that embed Kelly as a special case while capturing a spectrum of risk attitudes.

A notable example connecting the Kelly utility function to other forms of risk aversion is given in \cite{classicalswisspeople}, where the authors analyse Kelly gambling from the perspective of Constant Relative Risk Aversion (CRRA). CRRA describes a class of utility functions in which an individual's relative risk aversion remains constant regardless of their level of wealth. For example, if a person has more wealth, they might risk a larger absolute amount while keeping the proportion of wealth they are willing to risk constant. This property makes CRRA particularly suitable for analyzing the Kelly paradigm, as it aligns with the proportional nature of Kelly betting as well as the i.i.d.\ strategy assumption.The CRRA family of utility functions is given by: 
\begin{equation}
    u_\beta(w) = 
    \begin{cases}
      \dfrac{w^{\,1-\beta}-1}{1-\beta}, & \beta \neq 1,\\[1ex]
      \log w, & \beta = 1,
    \end{cases}
\end{equation}
where $\beta \ge 0$ is the relative risk aversion parameter.
The $\beta$ parameter
continuously interpolates between logarithmic utility ($\beta=1$) and other risk attitudes. 


\label{optimisers}
In \cite{classicalswisspeople}, 
the authors prove that Alice's optimal strategy for a particular value of $\beta$ is given by $Q_{X}^{A,\beta}$, where
\begin{equation}
Q_{X}^{A,\beta}(x) = \frac{P_X(x)^{\frac{1}{1-\beta}} O_X(x)^{\frac{\beta}{1-\beta}}}{\sum_{x'} P_X(x')^{\frac{1}{1-\beta}} O_X(x')^{\frac{\beta}{1-\beta}}},
\end{equation}
with \( O_X(x)^{-1} = Q_X^B(x) \) in our definition of odds. 
This result highlights the role of the parameter \( \beta \) in controlling the trade-off between risk and reward. Notably, when \( \beta \to 0 \), the utility function reduces to logarithmic utility, corresponding to the Kelly criterion. For other values of \( \beta \), the utility function captures different levels of risk aversion.

We now want to understand how the result 
of \cite{classicalswisspeople} relates to the framework described in \cref{finitekelly}, where we characterized an individual's preferences in terms of their ability to compromise between the probability of a type and the wealth obtained when the type matches the observed sequence. Noting that
\begin{equation}
    \frac{-\beta}{1-\beta} = 1 - \frac{1}{1-\beta},
\end{equation}
we see that the optimizer \( g^{(\beta)} \) found by \cite{classicalswisspeople} matches the optimizer in Equation \ref{OptimiserHT} after re-defining \( \eta = \frac{1}{1-\beta} \). This equivalence demonstrates that the risk-reward tradeoff presented in our information-theoretic framework is fundamentally connected to the utility maximization problem analysed in \cite{classicalswisspeople}. 

We note that though parts of the argument for risk-reward tradeoffs presented in the previous section relied on $n$ being `sufficiently large', the expected utility hypothesis itself relies on the notion of expected values, which are most meaningful in the limit of large $n$. In finite-
$n$ scenarios, the actual realized outcomes may deviate significantly from the expected value, making the expected utility hypothesis less applicable. This creates a tension when applying the hypothesis to single-shot or finite-horizon decision-making, as the gambler cannot rely on the law of large numbers to average out fluctuations. We also emphasise that equation \ref{singleshotKelly}  holds for any value of $n$. 

Using the optimal strategy, we show in \cref{ExpectedwealthAppendix} that the wealth for a rational individual may be written as
\begin{equation}
\boxed{
\log \bigg(  \frac{W_F}{W_i}  \bigg) = \alpha D(P_X || Q^B_X) + (1 - \alpha) D_\alpha(P_X || Q^B_X), \quad \text{where} \quad \alpha = \frac{1}{1-\beta}
} \label{eq:final_expected_work}
\end{equation}

Which shows that the wealth of a rational individual is parametrised by Rényi divergences.

While \cite{classicalswisspeople} reduced the CRRA optimization problem to expressions involving Rényi divergences, they did not evaluate the corresponding payoff. By substituting the optimiser back into the Kelly expression, we find that the expected wealth takes the remarkably simple form above: a convex combination of KL and Rényi divergences. This step shows that Rényi divergences do not merely appear as auxiliary quantities in the optimization, but directly parametrize the operational payoff of the optimal strategy of a rational individual. This contrasts with heuristic financial interpretations of Rényi divergences \cite{soklakov2020economics}. 

Our results highlight how tools from information theory (Rényi divergences) and concepts from economics (CRRA utilities) are two views of the same optimisation problem.




\section{Gambling with side information}\label{doublesideinfo}


The Kelly framework, generalised above to the finite-size regime, models a \textit{passive} adversary: Bob sets fixed odds, and Alice optimizes unilaterally. Although the framework possesses a fundamental resource-theoretic structure as described above, many adversarial scenarios involve \textit{active} opponents who strategically adapt to maximize their own advantage. We have also assumed that Alice knows the prior distribution $P_X$ which is drawn i.i.d.\ and we would like to consider the case where in each round, Alice has different information about the race outcomes.

We would like to extend the scenario we considered in Section \ref{modernkelly} to domains like cryptography, where adversaries (eavesdroppers) actively probe protocols to extract secrets~\cite{krawec2020gametheoryquantumcrypto}; economic competition, where traders and bidders dynamically counter opponents' strategies~\cite{vonStengelKoller1997}; and quantum protocols, where malicious parties exploit entanglement and measurements to gain advantage—all scenarios requiring active strategic interplay beyond passive constraints

In this section, we generalise Kelly gambling to a distributed adversarial game in which the players Alice (gambler), Bob (adversary), and Charlie (referee/source) observe correlated signals—$X$, $Y$, and $Z$ respectively—but lack knowledge of each other's information. Alice and Bob simultaneously commit to actions based solely on their private signals.

The tripartite structure, defined by a joint distribution $P_{XYZ}$, is essential for three reasons: it explicitly models asymmetric information (e.g., Alice observes $X$, Bob observes $Y$, and neither accesses $Z$ directly); it enforces strategic interdependence by making payoffs contingent on both players' actions simultaneously; and it enables quantum generalization, where $\rho_{ABC}$ naturally replaces $P_{XYZ}$ to model distributed quantum states. 

This setup, more appropriate to model real-world scenarios where decision-makers must act based on incomplete or noisy information, is explored in economics and game theory through the lens of games with \textit{incomplete information}. Though the betting game is zero-sum by assumption, and we later prove that asymptotically optimal strategies form a Nash\footnote{A \textit{Nash equilibrium} is a set of strategies where no player can gain a higher payoff by unilaterally changing their own strategy, given what the other players are doing. } equilibrium \textit{without} the rationality assumptions often imposed in game theory.

 For readers unfamiliar with game theory, we note that games with \textit{incomplete information} are games in which players are uncertain about some important parameters of the game situation, such as the payoffs and the strategies available to other players. Although the players are uncertain about these parameters, each player may have a subjective probability distribution over the alternative possibilities. In these settings, it is common \cite{harsanyi1967games1, Rasmusen2007} to assume that the distributions which each player has are `mutually consistent' in the sense that they can all be derived from a common probability distribution, which is known to all players. Harsanyi's foundational papers \cite{harsanyi1967games1} proved that, in these settings, Nash equilibria can be defined by introducing a \emph{delayed commitment} mechanism in which players first observe their private information before choosing strategies. This contrasts with classical complete-information games, where equilibria are specified directly from the payoff matrix without reference to any underlying information structure. 

When the consistency assumption holds such games can be transformed into an equivalent game of complete but \textit{imperfect information}, called the Bayes-equivalent of the original game. In the Bayes-equivalent of the original game, a fictitious player, often called `Nature', moves first, selecting the unknown parameters according to the common probability distribution, and then the game proceeds as a game in which players observe only their own private information but not that of others. This transformation allows the application of standard game-theoretic concepts, such as Nash equilibria, to analyse strategic interactions in settings with uncertainty.

\subsection{Gambling with incomplete information}

Let us consider a game involving three players: Alice, Bob, and Charlie; as in previous chapters. Alice and Bob each share correlations with Charlie's system, and their strategies depend on their respective side information. We imagine that the game is described by a tripartite probability distribution $P_{XYZ}$, and that a sequence $(x^{n},y^{n},z^{n})$ is jointly sampled from $P_{XYZ}$. Alice has access to $x^{n}$, whilst Bob has access to $y^{n}$, and they bet on the outcome $z^{n}$. We assume that the total distribution $P_{XYZ}$ from which Alice and Bob can calculate marginal distributions such as $P_{XZ}$, $P_{XY}$ and any other probabilities of interest whenever they observe realisations of their random variables is known to everyone. This is a standard assumption in game theory \cite{harsanyi1967games1}. We will continue to assume that the sequence  $(x^{n},y^{n},z^{n})$ is sampled i.i.d, and that Alice's and Bob's strategies reflect this. Unlike the standard Kelly scenario, we will no longer assume that Alice observes the odds before deciding on her strategy, in order to emphasise that the role of the adversary is no longer passive and that Alice has incomplete information about Bob's preferences. This will also allow easier generalisation to the quantum scenario in \cite{ARTArcos2025adversarial}. As in the previous section, our first focus is to understand what Alice and Bob think the probability of a sequence $z^{n}$ is given that they have observed the realisation of $x^{n}$ and $y^{n}$ respectively, as well as how they may allocate fractions of their wealth to different sequences based on their observations.

Following the methods used in the previous chapter, and in particular the method of types \cite{csiszar1998method}, we can find expressions for Alice's change in wealth as well as the probabilities of different sequences.
Suppose that $(x^{n},z^{n})$ are drawn $n$ times i.i.d. from the distribution $P_{XZ}$. We may first write the probability of the sequence $z^{n}$ given $x^{n}$ as
\begin{align}
    P_{XZ}(z^{n}|x^{n})&= \prod_{i} P(z_{i}|x_{i})
\nonumber\\
   & = \prod_{x \in \mathcal{X}, z \in \mathcal{Z}} P(z|x)^{N(x,z|x^{n},z^{n})}
\end{align}
where $N(x,z|x^{n},z^{n})$ is the number of times the pair $(x,z)$ occurs. Using that $N(x,z|x^{n},z^{n})= n \lambda_{x^{n},z^{n}}(x,z)$, we can write
\begin{equation}
    P_{XZ}(z^{n}|x^{n})= \prod_{(x, z) \in \mathcal{X} \times \mathcal{Z}} P(z|x)^{n \lambda_{x^{n},z^{n}}(x,z)}
\end{equation}
This can be written, as first shown by \cite{csiszar1998method}:
\begin{equation}\label{csiszartypes}
    P_{XZ}(z^{n}|x^{n})= 2^{-n(H(\lambda_{x^{n},z^{n}})-H(\lambda_{x^{n}})+ D(\lambda_{x^{n}z^{n}}\| W_{XZ} ))}
\end{equation}
where $W_{XZ}$ is the distribution defined by
\begin{equation}
    W_{XZ}(x,z)= \lambda_{x^{n}}(x)P_{XZ}(z|x)
\end{equation}

As in the previous chapter, Alice and Bob can use \cref{csiszartypes} to calculate the probability of the sequences $z^{n}$ given their side information and the overall known distribution.

Suppose now that Alice chooses to allocate fractions of her wealth to possible outcomes $z \in \mathcal{Z}$ whenever she receives $x \in \mathcal{X}$ according to fractions given by $Q^{A}(z|x)$. If Bob assigns odds according to $Q^{B}(z|y)$, the payoff received by Alice when the outcome $z^{n}$ is revealed is given by
\begin{align}\label{singleshotbayesian}
    \frac{Q^{A}(z^{n}|x^{n})}{Q^{B}(z^{n}|y^{n})}
    =&
     2^{n(D(\lambda_{y^{n}z^{n}}\| \tilde{W}_{YZ} )-D(\lambda_{x^{n}z^{n}}\| \tilde{W}_{XZ})  
    + H_\lambda(Z|Y)-H_\lambda(Z:X)
     )}
\end{align}
where $\tilde{W}_{XZ}(x,z) = \lambda_{x^{n}}(x)Q^{A}(z|x)$, $\tilde{W}_{YZ}(y,z)=\lambda_{y^{n}}(y)Q^{B}(z|y)$ are the distributions defined by Alice's and Bob's strategies respectively, and 
$H_\lambda(Z|Y):=H(\lambda_{y^{n},z^{n}})-H(\lambda_{y^{n}})$, $H_\lambda(Z|X):=H(\lambda_{x^{n},z^{n}})-H(\lambda_{x^{n}})$ can be thought of as a empirical mutual information and is independent of Alice and Bob's strategies. From the expression above, it is clear that Alice must choose a conditional strategy $Q^{A}(z|x)$ such that the product matches the \textit{joint} type. This expression generalises the single-shot Kelly result to the case where both gambler and adversary have access to side information. 

This ratio can also be expressed in a more elucidating form:
\begin{equation}\label{eq:GlobalTypeRelative}
\boxed{\frac{Q^A(z^{n}| x^{n})}{Q^B(z^{n}| y^{n})}
= 2^{n\big(D(\lambda_{x^{n} y^n z^{n}}\| \tilde{W}_{B})
-D(\lambda_{x^n y^{n} z^{n}}\| \tilde{W}_{A})
\big)}},
\end{equation}
where $\tilde{W}_A(x,y,z)=\lambda_{x^n y^n}(x,y)Q^A(z| x)$ and $\tilde{W}_{B}(x,y,z)=\lambda_{x^n y^n}(x,y)Q^B(z|y)$


Unlike the standard Kelly scenario—a one-sided optimization problem against fixed odds—we now model a game in which both Alice and Bob must reason about the strategies of one another. In particular, a risk-seeking Alice may choose to concentrate her bet on strings for which 
$D(\lambda_{x^{n} y^n z^{n}}\| \tilde{W}_{B})$
is high, which would guarantee higher payoffs. 
After she observes her $x^n$, this corresponds to optimising over type-classes $\lambda_{y^nz^n}(yz|x)$. In such a case, her probability of success is given in terms of $D(\lambda_{y^nz^n}(y,z|x))$ as in Eq. \eqref{eq:RiskConstraint}.
The shift in Bob's role transforms the interaction into a true game in the game-theoretic sense: Alice and Bob now engage in a strategic interplay where each player’s decisions directly constrain the other’s outcomes. We discuss this in more detail in \cite{ARTArcos2025adversarial}, as it has analogies in the quantum case.

\subsection{Asymptotic limits}

We began this section by observing that in the original Kelly framework, the adversary’s role was passive— Bob simply set fixed odds, and Alice optimised her utility function against her opponent. In this chapter, we introduced a more dynamic adversary: Bob now actively uses private side information to set odds, while Alice operates without knowing his strategy. Despite this added complexity, an analytic expression for Alice’s wealth rate, under only the assumption that Bob's odds and Alice's bets were consistent across sequences of the same type, can be derived. Equation \ref{singleshotbayesian} reveals that, just as single-shot Kelly betting reduced to Alice guessing the type of a sequence, Alice's role in a game with double side information reduces to inferring the joint type of correlated sequences. 

Just in the way in which expected utility theory analyses Alice's (the expected utility maximiser) behaviour by considering her optimisation of a utility function, game-theoretic analyses of zero-sum games focus on the role of \textit{Nash equilibria} of games. A \textit{Nash equilibrium} is a strategic configuration where no player can improve their payoff by unilaterally changing their strategy, assuming all other players keep theirs fixed. This concept extends to the type of games considered in this section, which are called games with \textit{imperfect information} in game theory.

Analysing expression for the rate at which Alice's wealth grows for finite $n$, and a full analysis of the Bayesian Nash equilibria of the game in this scenario is beyond the scope of the current work, and is left as an open problem for future research. Nonetheless, we consider the limit in which $n \to \infty$, in which we are able to analyse Nash equilibria. Note first that, in this limit, all empirical sequences reflect the distribution of the source, and Alice's wealth grows as Eq. \eqref{eq:GlobalTypeRelative}


Here, it is clear that Alice's and Bob's optimal strategies are to set their bets as $Q^{A}(z|x)=P_{Z|X}(z|x)$ and $Q^{B}(z|y)=P_{Z|Y}(z|y)$ respectively. This is due to the law of large numbers, which guarantees that all types reflect the distribution of the source. It is also clear that no player can improve their payoff by unilaterally changing their strategy, since this would introduce a divergence penalty that reduces the deviating player's payoff. When Alice and Bob play according to these strategies, therefore, Alice's payoff (called the \textit{value} of the game in game theory), is given by 

\begin{equation}
   \boxed{ \frac{Q^{A}(z^{n}|x^{n})}{Q^{B}(z^{n}|y^{n})}= 2^{n(H(Z|Y)-H(Z|X))} }
   \label{eq:asymmtotic}
\end{equation}
which quantifies the asymmetry in Bob and Alice's information. 


\paragraph{Remark}  
The derived payoff duality \emph{suggests} an underlying game structure:  
\begin{itemize}  
\item The minimax value aligns with Nash equilibrium payoffs in classical games.  
\item The optimal strategies ($P_{Z|X}$, $P_{Z|Y}$) coincide with Bayes-rational players.  
\end{itemize}  
Thus, information-theoretic structures \emph{induce} game-theoretic behaviour asymptotically, even without explicit rationality axioms.  In contrast, our formulation does not posit any such common prior or
rationality axioms. Instead, the equilibrium behaviour emerges
\emph{solely} from the information structure: the fact that Alice and Bob
must commit strategies based only on their private signals, and that
their payoffs are constrained by the joint distribution $P_{XYZ}$. In
other words, where Harsanyi derives equilibria by embedding incomplete
information into a fictitious complete-information game with Nature as a
first mover, we show that the same equilibrium structure can be recovered
from the statistical correlations themselves. This shift highlights the
information-theoretic perspective: information asymmetry, rather than
assumptions of rationality, is the fundamental resource that generates
strategic equilibria.

\section{The resource theory of gambling}
\label{sec:GamblingRT}


The term \textit{resource theories} (RTs) refers to a framework which systematically allows the quantification of the resources required to carry out information processing tasks \cite{Oppenheim2010Quantumness,brandao_resource_2015,gour_chitambar_art_2018, gour2025quantum}. The formalism has provided insight into various aspects of physics, including computation, non-equilibrium thermodynamics, entanglement theory, quantum coherence, and quantum physics under symmetric restrictions; offering a unifying approach to understand the limits of information processing in physical systems.

The basic structure of resource theories relies on defining some restricted class $\mathcal{A}$ of \textit{allowed operations} and 
the set $\mathcal{P}$ of allowed states they operate on. 
Defining the set of allowed operations allows one to identify which states are more resourceful: we can say that a state is more resourceful than another if an allowed operation in $\mathcal{A}$  maps the first state to the second. This is because an agent who is restricted to operations from $\mathcal{A}$ would clearly prefer to have the first state over the second state since they can always transition the second state using the first state as an input. At the extreme end, there are states which can be created under the allowed class of operations without requiring the input of any resourceful state -- those are refered to as the set $\mathcal{F} \subseteq \mathcal{P}$ of \textit{free states}. The states which cannot be created under the allowed class of operations for free are a resource.

Before specifying the allowed operations in the resource theory of gambling, it is helpful to motivate their choice. In any resource theory, the resources are those features that cannot be generated under the permitted class of operations. If we wish to consider the resource theory of correlations, the natural class of operations is Local Operations (LO)\cite{HendersonVedral2001}: local processing at Alice, Bob, and Charlie cannot create correlation, hence any observed correlation must be treated as a resource. If Alice and Bob have prior information about the outcome of the race, this is captured by their side-information being correlated with the outcomes $Z$. This provides the rationale for adopting LO as the baseline in our gambling RT. However, in the gambling context, one must also allow communication. 
This captures the fact that, even over $n$ rounds, a gambler’s
payoff can be superlinear in $n$ on rare events (e.g. they could make a long-shot bet and win an arbitrary amount). This
cannot be modelled by Local Operations on some finite sized state e.g. $P_{XYZ}^{\otimes n}$, because one cannot create a superlinear amount of correlation by LO from  $P_{XYZ}^{\otimes n}$.
. 

We will now see that if we formalise a resource theory (RT) whose underlying resource is the ability of
Charlie to deliver classical communication to Alice, then this precisely captures the mathematical structure of gambling, as laid out in the preceding sections.
 
\paragraph{Systems and states.}
The basic system consists of three finite alphabets \(\mathcal X,\mathcal Y,\mathcal Z\).
A \emph{state} is any joint probability distribution \(P_{XYZ}\) on
\(\mathcal X\times\mathcal Y\times\mathcal Z\).
(When we study blocklength \(n\), the state is \(P_{X^nY^nZ^n}\); i.i.d.\ is a
special case with \(P_{XYZ}^{\otimes n}\).) Here, $Z$ are the outcomes of the event which is being bet on (e.g. horse races), and $X$ and $Y$ are Alice and Bob's side information.

\paragraph{Allowed operations.}
Free operations are \emph{local stochastic maps} on Alice and Bob. We will also allow
the \emph{use of a classical channel \(C\!\to\!A\)} whose cost is the number of uses of that channel. The number of uses of the channel may depend on the realised \(z\). Formally, an allowed operation from
\(P_{XYZ}\) to \(P'_{X'Y'Z'}\) consists of:
\begin{enumerate}
	\item Alice: a stochastic map \(T_A:x\mapsto x'\) (possibly using local
	randomness). If the output alphabet is the same as that of $Z$, then this models a bet.
	\item Bob: a stochastic map \(T_B:y\mapsto y'\).
     If the output alphabet is the same as that of $Z$, then this models setting the odds.
	\item Charlie: a deterministic relabelling \(\pi_C:z\mapsto z'\). In the more general setting we could expand this to any stochastic map, but we do not allow this here.
    \item \label{op:cc} Charlie: 
	a classical message \(m\in\{0,1\}^\ell\) sent to Alice (or Bob), whose cost is given by the length of the message
	\(\ell=\ell(z)\) which may depend on \(z\).
\end{enumerate}
No other communication is permitted.  Crucially, the number of bits
	sent from \(C\) to \(A\) in operation \ref{op:cc} is the resource we account for; it is \emph{not} a
free operation, but we do need to allow and account for it.
In this way, the RT of gambling extends the standard RT framework, in that not only are there resourceful states, but also resourceful operations.

\paragraph{Free states.}
Free states are exactly those that can be prepared \emph{without} using any
\(C\!\to\!A\) communication, i.e. the uncorrelated distribution
\[
\mathfrak F \;=\; \bigl\{\,P_{XYZ}=P_XP_YP_Z\,\bigr\}.
\]
Indeed, any product \(P_XP_YP_Z\) can be created by local sampling at
\(A,B,C\); conversely, if \(X,Y,Z\) are not mutually independent, generating
their correlations requires (on average) a positive amount of communication,
hence they are not free.

\paragraph{RT monotones.}
A (state) monotone is any functional $M$ such that
$M(P_{XYZ}) \ge M(P'_{X'Y'Z'})$ for all allowed operations which takes $P_{XYZ}\rightarrow P'_{X'Y'Z'}$. Monotones therefore provide a partial ordering which quantify how valuable different probability distributions are. 
Some natural families of monotones under LO are:
\begin{align}
    M_\alpha
    (P_{XYZ})
    &:= \inf_{Q\in\mathfrak F} D_\alpha\!\left(P_{XYZ}\,\|\,Q\right), \\
    E_\alpha(X:Z)
    &:= \inf_{Q_X} D_\alpha\!\left(P_{XZ}\,\|\,Q_X \otimes P_Z\right), \\
    E_\alpha(Y:Z)
    &:= \inf_{Q_Y} D_\alpha\!\left(P_{YZ}\,\|\,Q_Y \otimes P_Z\right), \\
    E_\alpha(XY:Z)
    &:= \inf_{Q_X \otimes Q_Y} D_\alpha\!\left(P_{XYZ}\,\|\,Q_X \otimes Q_Y \otimes P_Z\right).
\end{align}
where $D_\alpha$ is any data-processing–contractive divergence i.e. $D_\alpha(T(P)||T(Q)) \leq D_\alpha(P||Q)$ for any channel $T$, and we use the $\otimes$ notation to indicate a product probability distribution for easy generalisation to the quantum case in \cite{ARTArcos2025adversarial}.
The Rényi divergences are typically used as contractive distances~\cite{secondlawspaper} in this context with the KL divergence ($\alpha=1$) playing a special role as a unique monotone for reversible resource theories\cite{isentanglementthermo,Oppenheim2010Quantumness,brandao_resource_2015}.
Here, this role is played by the mutual informations, which can be written in terms of the KL divergence as $I(X:Z)=E_1(X:Z)$\cite{CoverThomas2006} and likewise for $I(XY:X)$ and $I(Y:Z)$. In the limit of large $n$ the RT theory becomes asymptotically reversible and we will recover the mutual information as monotones as in Eqn \eqref{eq:asymmtotic}.
The $E_\alpha$ functionals are the classical analogues of correlation monotones familiar from quantum resource theories: they measure the minimal divergence between the actual joint distribution and one where Alice and/or Bob are uncorrelated with the race outcome $Z$. 

Since we do not allow arbitrary channels on $C$, but only relabelling, the following generalisations of conditional entropies (or more accurately, {\it conditional negentropies}) are also monotones in the resource theory
\begin{align}
    E_\alpha(Z|X)
    &:= \inf_{Q_X} D_\alpha\!\left(P_{XZ}\,\|\,Q_X \otimes \bone_Z\right), \\
    E_\alpha(Z|Y)
    &:= \inf_{Q_Y} D_\alpha\!\left(P_{YZ}\,\|\,Q_Y \otimes \bone_Z\right), \\
    E_\alpha(Z|XY)
    &:= \inf_{Q_X \otimes Q_Y} D_\alpha\!\left(P_{XYZ}\,\|\,Q_X \otimes Q_Y \otimes\bone_Z\right).
\end{align}
where $\bone_Z$ is $1$ for all $z$, and we once again use  notation borrowed from quantum theory. $E_1(Z|X)=\log|\mathcal{Z}|-H(Z|X)$, which gives an indication of why we call these conditional negentropies. It is $0$ when $Z$ is uniformly distributed, even when conditioned on $X$, and maximal when $Z$ is perfectly correlated with $X$. 

\paragraph{Adversarial resource quantifiers (ARQs).}
In addition to state monotones, we often care about task-specific payoff functionals
that are optimised by the parties\cite{prehistoricART}. We will call these adversarial resource quantifiers\cite{ARTArcos2025adversarial}, because they can only increase under the action of the protagonist Alice, and only go down under the action of Bob the adversary. For completeness, we include a brief discussion of them here.

Given a convex $f:\mathbb R_+\!\to\!\mathbb R$, we define the ARQ in~\cite{ARTArcos2025adversarial}
\begin{equation}
	\label{eq:ARQ-def}
	\RQ_f^{(n)}(P_{X^nY^nZ^n})
	:= \sup_{Q_A}\inf_{Q_B}
	\,\mathbb E\!\left[
	f\!\left(\frac{Q_A(Z^n|X^n)}{Q_B(Z^n|Y^n)}\right)\right],
\end{equation}
where $Q_A(\cdot|x^n)$ and $Q_B(\cdot|y^n)$ are conditional pmfs (channels)
selected by Alice and Bob.  These are the bets and odds from Section \ref{{doublesideinfo}}.
The ratio $\tfrac{Q_A(Z^n|X^n)}{Q_B(Z^n|Y^n)}$ is thus the \emph{wealth relative} for a given realisation, so Alice’s winnings are a random variable determined by her and Bob’s strategies.  
The functional $\RQ_f^{(n)}$ evaluates this random wealth through the convex scoring function $f$.  
For the special case $f=\log$, we obtain
\[
X_{\log}^{(n)}(P_{X^nY^nZ^n})
= \sup_{Q_A}\inf_{Q_B}\,\mathbb E\!\left[\log\!\left(\tfrac{Q_A}{Q_B}\right)\right],
\]
which is exactly the expected \emph{wealth relative}, i.e.\ the Kelly criterion payoff.  
For general convex $f$, the ARQ instead evaluates the distribution of winnings through $f$, thereby capturing alternative convex risk measures of Alice’s advantage.  
Thus ARQs provide a family of adversarially defined payoff quantifiers, distinct from the RT monotones.

\paragraph{Gambling game realisation (within the RT).}
We now wish to specify what is the equivalent of a gambling game within the context of this resource theory. We require that there is an exact correspondence between the wealth relative in gambling, and the resource in the RT (in this case, correlation or communication of information). 
In this model, wealth is replaced by a physical resource: communication bits. This allows us to see how the rules of betting emerge from the physical constraints of transmitting information. We will also see that gambling can be interpreted purely in terms of information. The amount that a party wins, can be quantified in terms of the difference in  information that each party has about the race outcomes. And this in turn is quantified in terms of the amount of communication that would be required for the parties to learn the race outcome.

The game can be seen as a competition between Alice and Bob to see who can more efficiently compress information about the race outcome, $z^n$. The player with the better compression scheme wins the difference in efficiency as a prize.

\begin{enumerate}
    \item \textbf{The Odds as a Communication Promise:} We interpret Bob's odds as a public commitment to a variable length code, and in particular, a string length $\ell_B(z^n|y^n)$ for each message $z^n$. The analogue of commiting to a  payout to Alice for each possible sequence of race outcomes $z^n$, is that he commits that $\ell_B(z^n|y^n)$ bits will be communicated from Charlie to the winner (Alice) for each possible $z^n$. A reasonable strategy for him would be to commit to sending more bits for outcomes he thinks are unlikely (long odds). He should commit to a sending a smaller number of bits for a likely outcome (short odds). To be a valid, communication strategy that enables Alice to decode the message with certainty, his communication plan must satisfy the Kraft-McMillan inequality, $\sum 2^{-\ell_B} \le 1$. This is the condition for a code to be prefix-free, which is what is required if Alice is to be able to decode $z^n$ unambiguously. If we demand that the book keeper specify a pay-out for every possible $z^n$, then the condition  $\sum 2^{-\ell_B} = 1$ must be satisfied.
  
    \item \textbf{The Bet as a Communication Cost:} Alice's bet is her own declaration of how many bits are truly necessary to identify the outcome $z^n$, given her private information. She commits to a variable length code with costs of $\ell_A(z^n|x^n)$ bits for each outcome. This plan must also be self-consistent and complete ($\sum 2^{-\ell_A} =1$). 
     \item \textbf{The Payoff as Surplus Communication:} When a specific outcome $z^n$ occurs, the game plays out as follows:
        \begin{itemize}
            \item Bob's commitment is triggered, making a channel of capacity $\ell_B(z^n|y^n)$ bits available.
            \item Alice uses $\ell_A(z^n|x^n)$ bits of this capacity to have the outcome communicated to her. 
               \item If $\ell_B \ge \ell_A$, the communication succeeds in Alice learning $z^n$. The leftover channel capacity, $k = \ell_B(z^n|y^n) - \ell_A(z^n|x^n)$, is her winnings. She can immediately use these $k$ surplus bits to receive extra, valuable information (a "jackpot") from the referee, Charlie. This can be formalised by introducing by an auxillary random variable $\mathcal{Z}'$ at Charlie's site, which Alice would also like to learn.
               \item If instead $\ell_B \le \ell_A$, then Alice will need to pay Bob for additional communication from Charlie in order for her to learn $z^n$. In this case, $k$ is negative and represents her cost.
            \end{itemize}
    \end{enumerate}

\noindent Remark: The requirement that their codes be complete and prefix-free is what guarantees that $z^n$ is always uniquely decodable. 
Their codes thus satisfies the Kraft–McMillan equality
\[
\sum_{z^n}2^{-\ell_A(z^n|x^n)}=1,\qquad
\sum_{z^n}2^{-\ell_B(z^n|y^n)}=1.
\]
Identifying \(Q_A:=2^{-\ell_A}\), \(Q_B:=2^{-\ell_B}\), the payout satisfies
\(2^{k}=\bigl(Q_A/Q_B\bigr)\) on winning realisations, reproducing the Kelly
wealth ratio per type.  In the i.i.d.\ regime, optimal average lengths satisfy
\(\frac1n\mathbb E[\ell_A]\to H(Z|X)\) and
\(\frac1n\mathbb E[\ell_B]\to H(Z|Y)\), so the first-order expected payout rate
is \(H(Z|Y)-H(Z|X)=I(X\!:\!Z)-I(Y\!:\!Z)\), consistent with the asymptotics, and with the KL divergence monotones given above.

\section{Conclusion}

In this work, we have reframed the classical theory of gambling as a resource theory of adversarial information. By shifting the focus from an unobservable ``true" distribution to the empirical distribution of finite outcomes, we extended Kelly's seminal work to the single-shot and finite-n regimes. This approach revealed a deep connection between three seemingly disparate fields: the information-theoretic analysis of risk via the method of types, the economic theory of rational choice under uncertainty via CRRA utility functions, and the statistical problem of asymmetric hypothesis testing. We showed that these are not merely analogous but are mathematically equivalent descriptions of the same underlying risk-reward trade-off, governed by the family of Rényi divergences.

Our generalisation to tripartite games with side information demonstrates that Nash equilibria can emerge asymptotically from purely information-theoretic constraints, without recourse to rationality axioms. This suggests that strategic behaviour may arise naturally from optimal information processing rather than requiring separate rational choice principles.

An outcome of this investigation is a fully operational resource theory of gambling. We found, perhaps surprisingly, that a consistent theory of gambling cannot simply be a resource theory of correlations. The possibility of super-linear wins requires a more powerful operational primitive: state-dependent communication. By identifying the gambler's payoff with a surplus of communication bits from a referee, we arrived at a framework where the mathematical structure of betting—specifically, the need for strategies to be probability distributions satisfying Kraft's inequality—emerges as a physical necessity for designing a complete and unambiguous communication protocol. The resource is correlation, the actions are local choices of compression schemes, and the payoff is communication.

This operational resource-theoretic framework provides a robust and natural pathway to the quantum domain. While this paper has focused on the classical foundations, the framework presented here is ripe for generalisation. In our companion work~\cite{ARTArcos2025adversarial}, we develop the broader structure of Adversarial Quantum Resource Theories (AQRTs). There, the classical state $P_{XYZ}$ is replaced by a tripartite quantum state $\rho_{ABC}$.
The results found here allow us to apply the lessons here to a variety of other settings. One notable example is thermodynamics, where we find that the risk-reward trade-off developed here, corresponds to a risk-reward trade-off in work extraction from small systems\cite{AdversarialThermo}.

In the quantum analogue of gambling, we propose in \cite{ARTArcos2025adversarial} that classical channels are replaced by quantum ones, and information transfer is replaced by quantum state transfer. A quantum version of Eqn \eqref{eq:ARQ-def} acts as a quantum adversarial quantifier, to capture the payoff.

In this quantum gambling game, the operational meaning of winning a bet becomes winning quantum communication rather than classical communication. In this way, the quantum analogue of our communication game is a form of adversarial state merging using variable length quantum codes. Charlie is able to send his quantum system to the gambler, against the noise from Bob's system. The payoff is a net gain or loss of entanglement. This puts us in a good position to explore other quantum adversarial resources theories. 

We also find that the one-shot quantum gambling game exhibits a risk-reward tradeoff governed not by classical divergences, but by the quantum hypothesis testing entropy. This demonstrates that the deep connection between gambling and hypothesis testing that we found here, persists and is even enriched in the quantum world. For convenience, a brief sketch of these result can be found in the appendix.

By developing the classical foundation of gambling, we aim to build a bridge that allows the powerful tools of economic decision theory and the concrete scenarios of gambling to be applied to other domains, including quantum information theory. Our results open the door to quantifying what it truly means for one quantum agent to ``know more" than another, a question that lies at the heart of quantum cryptography, quantum communication, and the acquisition of information in the quantum world.

\pagebreak

\section*{Acknowledgements}

We thank Takahiro Sagawa for useful discussions and feedback on the draft. 
MA thanks Francisco Aristi Reina, Galit Ashkenazi-Golan, and Robert Salzmann for valuable discussions. MA was supported by the Mexican
National Council of Science and Technology (CONACYT), EPSRC, the Simons Foundation It from Qubit Network, and SNF Quantum Flagship Replacement Scheme (Grant No. 215933). TS is supported by by JST ERATO Grant Number JPMJER2302, Japan. RR acknowledges support from the National Centre of Competence in Research SwissMAP and the ETH Zurich Quantum Center. 

\appendix



\section{Proof of optimal strategies}\label{lagrangemultipliers}

The aim of this appendix is to provide a step–by–step derivation of the optimisers
for the constrained divergence problems introduced in  \cref{finitekelly}. Though the solutions are well-known in information theory we include them here for completeness. 

In the main text we argued that wealth maximisation can be phrased either as
(i) maximising the divergence to $Q^B_X$ subject to a constraint on the
divergence to $P_X$, or equivalently (ii) minimising the divergence to $P_X$
subject to a constraint on the divergence to $Q^B_X$. Both formulations are convex
optimisation problems on the probability simplex, and both admit solutions lying on the same
family of distributions of the form:
\begin{equation}
Q^{A,\eta}_X(x) \;\propto\; P_X(x)^{1-\eta} Q^B_X(x)^{\eta}, \qquad \eta \in \mathbb{R}.
\end{equation}
What differs between the two problems is only the admissible branch of this
family: the Karush-Kuhn-Tucker (KKT) sign condition restricts $\eta$ to $[0,1]$ in case (i) and to
$\eta \le 0$ in case (ii). We now present the full derivations of both problems.

\begin{lemma}[Probability-constrained payoff maximisation]
Let $P_X$, $Q_X^A$, and $Q_X^B$ be probability distributions on a finite alphabet $\mathcal{X}$. For $d\ge 0$, consider
\[
\max_{Q^{A}_X}\; D(Q^A_X\|Q^B_X)
\quad\text{s.t.}\quad D(Q^A_X\|P_X)\le d,\ \sum_x Q^A_X(x)=1,\ Q^A_X(x)\ge 0.
\]
If the feasible set is nonempty, any optimizer has the form
\[
Q^{A*}_X(x)\;=\;\frac{P_X(x)^{1-\eta}\,Q_X^B(x)^{\eta}}{\sum_y P_X(y)^{1-\eta}\,Q^B_X(y)^{\eta}},
\qquad \eta\in[0,1],
\]
with $\eta=\lambda/(1+\lambda)$ where $\lambda\ge 0$ is the KKT multiplier; $\eta$ is chosen so that $D(Q^{A,*}_X\|P_X)=d$ when the constraint is satisfied
\end{lemma}

\begin{proof}
Write the inequality as $g(Q^A_{X})\equiv D(Q^A_{X}\|P)-d\le 0$ and form
\[
\mathcal L(Q^A_X,\alpha,\lambda)
= D(Q^A_X\|Q^B_X)\;+\;\lambda\,[\,d-D(Q^A_X\|P_X)\,]\;+\;\alpha\Big(\sum_x Q^A_X(x)-1\Big),
\quad \lambda\ge 0.
\]
Using $\partial_{q}\big(q\log\frac{q}{r}\big)=\log\frac{q}{r}+1$, stationarity for $Q^A_X(x)>0$ gives
\[
\Big(\log\tfrac{Q^A_X(x)}{Q^B_X(x)}+1\Big)\;-\;\lambda\Big(\log\tfrac{Q^A_X(x)}{P_X(x)}+1\Big)\;+\;\alpha=0.
\]
Rearranging,
\[
(1-\lambda)\log Q_X^A(x)=\log Q^B_X(x)-\lambda\log P_X(x)-(1-\lambda)(1+\alpha).
\]

Exponentiating and absorbing constants into the normaliser,
\[
Q^A_X(x)\ \propto\ Q^B_X(x)^{\tfrac{1}{1-\lambda}}\;P_X(x)^{-\tfrac{\lambda}{1-\lambda}}
\ =\ P_X(x)^{\,1-\eta}\,Q^B_X(x)^{\,\eta},\qquad
\eta:=\frac{\lambda}{1+\lambda}\in[0,1).
\]
Normalise to obtain the stated $Q^{A,*}_X$. Complementary slackness enforces $D(Q^{A,*}_X\|P)=d$ when the constraint is active; otherwise $\lambda=0$ and $\eta=0$.
\end{proof}

\begin{lemma}[Payoff-constrained probability maximisation]
Fix $P_X,Q^B_X$ and $k\ge 0$. Consider
\[
\min_{Q^{A}_X}\; D(Q^{A}_X\|P_X)
\quad\text{s.t.}\quad D(Q^A_X\|Q^B_X)\ge k,\ \sum_x Q^A_X(x)=1,\ Q^A_X(x)\ge 0.
\]
If the feasible set is nonempty, any optimizer has the form
\[
Q^{A,*}_X\;=\;\frac{P_X(x)^{1-\eta}\,Q^B_X(x)^{\eta}}{\sum_y P_X(y)^{1-\eta}\,Q^B_X(y)^{\eta}},
\qquad \eta\le 0,
\]
with $\eta=-\lambda/(1-\lambda)$ where $\lambda\ge 0$ is the KKT multiplier; $\eta$ is chosen so that $D(Q_A^\star\|Q_B)=k$ when the constraint is active.
\end{lemma}

\begin{proof}
Encode the inequality as $g(Q^A_{X})\equiv d - D(Q^A_{X}\|Q^B_{X})\le 0$ (with $d=k$) and form
\[
\mathcal L(Q_A,\alpha,\lambda)
= D(Q^A_X\|P_X)\;+\;\alpha\Big(\sum_x Q^A_X(x)-1\Big)\;+\;\lambda\,[\,d-D(Q^A_X\|Q^B_X)\,],
\quad \lambda\ge 0.
\]
Stationarity yields
\[
\Big(\log\tfrac{Q^A_X(x)}{P_X(x)}+1\Big)\;-\;\lambda\Big(\log\tfrac{Q^A_X(x)}{Q^B_X(x)}+1\Big)\;+\;\alpha=0,
\]
so
\[
(1-\lambda)\log Q^A_X(x)=\log P_X(x)-\lambda\log Q^B_X(x)-(1-\lambda)(1+\alpha).
\]
Hence
\[
Q^A_X(x)\ \propto\ P_X(x)^{\tfrac{1}{1-\lambda}}\;Q^B_X(x)^{-\tfrac{\lambda}{1-\lambda}}
\ =\ P_X(x)^{\,1-\eta}\,Q^B_X(x)^{\,\eta},\qquad
\eta:=-\frac{\lambda}{1-\lambda}\le 0.
\]
Normalisation gives $Q^{A*}_X$. By complementary slackness, either $\lambda=0$ (constraint inactive, $Q^{A}_{X}=P_X$) or $\lambda>0$ and $D(Q^A_X\|Q^B_X)=k$.
\end{proof}

\section{Expected wealth as a function of strategy}\label{ExpectedwealthAppendix}

The wealth growth of a gambler with betting strategy $Q_{X}^{A}$ is known to be given by 

\begin{equation}
\mathbb{E}[W] = D(P_X || Q^B_X) - D(P_X || Q^A_X) \label{eq:expected_work_start}
\end{equation}

in the asymptotic limit, where $Q^B_X$ determines the odds. 

Our goal is to express this for the optimal strategy $Q_{X}^{A}$ which from now on we write as $Q^{A,r}_X$, to emphasise dependence on the risk parameter of the CARA utility function. We substitute the explicit form of the optimal strategy into the second term.

The optimal strategy for a given risk parameter $\beta$ is:

\begin{equation}\label{eq:optimal_strategy}
Q^{A,\beta}_X(x) = \frac{P_X(x)^{\frac{1}{1-\beta}} Q^B_X(x)^{\frac{-\beta}{1-\beta}}}{Z}, 
\end{equation}

where

\begin{equation}
     Z(\alpha) = \sum_{x'} P_X(x')^{\frac{1}{1-\beta}} Q^B_X(x')^{\frac{-\beta}{1-\beta}} \label{eq:optimal_strategy}
\end{equation}

Let $\alpha = \frac{1}{1-\beta}$ for notational clarity. Since the first divergence is independent of the strategy  $Q^{A,\beta}_X$, we compute the second divergence term, 
\begin{align}\label{eq:intermediate_div}
D(P_X || Q^{A,\beta}_X) &= \sum_x P_X(x) \log \frac{P_X(x)}{Q^{A,\beta}_X(x)} \\
&= \sum_x P_X(x) \left[ \log P_X(x) - \log \left( \frac{P_X(x)^\alpha Q^B_X(x)^{1-\alpha}}{Z} \right) \right] \\
&= \sum_x P_X(x) \left[ \log P_X(x) - \alpha \log P_X(x) - (1-\alpha) \log Q^B_X(x) + \log Z \right] \\
&= (1 - \alpha) \sum_x P_X(x) \log \frac{P_X(x)}{Q^B_X(x)} + \log Z \\
&= (1 - \alpha) D(P_X || Q^B_X) + \log Z(\alpha) 
\end{align}
The normalization constant $Z$ is related to the sum that defines the Rényi divergence. Recall its definition:
\begin{equation}
D_\alpha(P_X || Q^B_X) = \frac{1}{\alpha - 1} \log \sum_x P_X(x)^\alpha Q^B_X(x)^{1-\alpha} = \frac{1}{\alpha - 1} \log Z
\end{equation}
So that 
\begin{equation}
\log Z(\alpha) = (\alpha - 1) D_\alpha(P_X || Q^B_X) \label{eq:logZ}
\end{equation}
Substituting \eqref{eq:logZ} back into \eqref{eq:intermediate_div}:
\begin{align}
D(P_X || Q^{A,\beta}X) &= (1 - \alpha) D(P_X | Q^B_X) + (\alpha - 1) D_{\alpha}(P_X || Q^B_X) \\
&= (1 - \alpha) \left[ D(P_X || Q^B_X) - D_\alpha(P_X || Q^B_X) \right]
\end{align}
Finally, we substitute this result back into the original expression for expected wealth \eqref{eq:expected_work_start}:

\begin{align}
    \mathbb{E}[W] &= D(P_X || Q^B_X) - D(P_X || Q^{A,r}X) \\
    &= D(P_X || Q^B_X) - (1 - \alpha) \left[ D(P_X || Q^B_X) - D_{\alpha}(P_X || Q^B_X) \right] \\
    &= D(P_X || Q^B_X) - (1 - \alpha) D(P_X || Q^B_X) + (1 - \alpha) D_\alpha(P_X || Q^B_X) \\
    & = \alpha D(P_X || Q^B_X) + (1 - \alpha) D_\alpha(P_X || Q^B_X)
\end{align}
%
%
Thus, the expected wealth growth for the optimal strategy corresponding to risk parameter $r$ is:
\begin{equation}
\boxed{
\mathbb{E}[W] = \alpha D(P_X || Q^B_X) + (1 - \alpha) D_\alpha(P_X || Q^B_X), \quad \text{where} \quad \alpha = \frac{1}{1-\beta}
} \label{eq:final_expected_work}
\end{equation}

\section{Geodesics and risk-reward trade-offs}
\label{sec:thesimplex}

It will be useful to consider flat space, and to go from the standard metric 

\begin{equation}
    ds^{2}= dx_{0}^{2}+dx_{1}^{2}
\end{equation}

to the metric 

\begin{equation}
    ds^{2}= \frac{dy_{0}^{2}}{y_{0}^{2}} + \frac{dy_{1}^{2}}{y_{1}^{2}}
\end{equation}

via the change of coordinates $y_{i}=\ln(x_{i})$. In the new metric, the geodesic equation becomes 

\begin{equation}
    \ddot{y}_{i}= \frac{\dot{y}_{i}^{2}}{y_{i}}
\end{equation}

with boundary conditions $y_{i}(0)=p_{i}$ and $y_{i}(1)=q_{i}$, the equation has solution 

\begin{equation}
    y_{i}(\lambda)= p_{i}\exp(-\lambda \ln{p_{i}/q_{i}})
\end{equation}

and so the region of the probability simplex through the points $(p_{0},p_{1})$ and $(q_{0},q_{1})$ can be parametrised as 

\begin{equation}
    \gamma(\lambda)= \frac{1}{s}( p_{0}^{1-\lambda} q_{0}^{\lambda},  p_{1}^{1-\lambda} q_{1}^{\lambda} )
\end{equation}

for $\lambda \in (0,1)$, where 

\begin{equation}
    s= \sum_{i}p_{i}^{1-\lambda}q_{i}^{\lambda}
\end{equation}

Since any distribution $(r_{0},r_{1})$ on the line satisfies 

\begin{equation}
    r_{i}(\lambda)= \frac{1}{s}p_{i}^{1-\lambda}q_{i}^{\lambda}
\end{equation}

for any $\lambda \in (0,1)$, one has the equations 

\begin{equation}
    \ln( \frac{r_{i}}{p_{i}} ) = \ln(\frac{1}{s}) + \lambda \ln(\frac{q_{i}}{p_{i}})
\end{equation}
and
\begin{equation}
    \ln( \frac{r_{i}}{q_{i}} ) = \ln(\frac{1}{s}) + (\lambda-1) \ln(\frac{q_{i}}{p_{i}})
\end{equation}
which imply 
\begin{equation}
    D(R^{\lambda}||P)= \ln(\frac{1}{s}) +  \lambda \sum_{i} r_{i} \ln(\frac{q_{i}}{p_{i}})
\end{equation}

and 

\begin{equation}\label{lambdavaries}
    (1-\lambda) D(R||P)+ \lambda D(R||Q) = \ln(\frac{1}{s})
\end{equation}



Note that for $\lambda>0$ 

\begin{equation}
    \ln(s)= (\lambda-1) D_{\lambda}(q||p)
\end{equation}

 and hence equation \ref{lambdavaries} can be written 

 \begin{equation}\label{lambdavaries}
    (1-\lambda) D(R||P)+ \lambda D(R||Q) = (1-\lambda) D_{\lambda}(Q||P)
\end{equation}

\section{Bets and odds as stake vectors and probability distributions}
\label{sec:stakes}
In Kelly betting, Bob, given his side information $y^n$ sets the odds for each outcome $z^n$, while conditioned on $x^n$ Alice may bet on each outcome. While this assigns some weight to each outcome $z^n$ if  the odds are superfair (sum to $1$ and if Alice stakes all of her portfolio on the outcome of $z$ then we may think of the odds and bet as a probability distribution as follows.
Given a realisation $x\in\mathcal{X}$, Alice chooses a \emph{stake vector}
\[
\mathbf{\bet}(x) = \big(\bet_z(x)\big)_{z\in\mathcal{Z}}, \quad \bet_z(x) \ge 0, \quad \sum_{z\in\mathcal{Z}} \bet_z(x) = 1,
\]
where $\bet_z(x)$ is the fraction of her wealth she allocates to outcome $z$ on observing $x$. 
Similarly, given $y\in\mathcal{Y}$, Bob chooses an \emph{odds vector}
\[
\mathbf{o}(y) = \big(o_z(y)\big)_{z\in\mathcal{Z}}, \quad o_z(y) \ge 0, \quad \sum_{z\in\mathcal{Z}} o_z(y) = 1.
\]
Although operationally $\mathbf{\bet}(x)$ and $\mathbf{o}(y)$ are portfolios of stakes and odds, mathematically they are points in the simplex $\Delta(\mathcal{Z})$ and can be identified with conditional probability distributions
\[
q(z|x) := \bet_z(x), 
\qquad r(z|y) := o_z(y).
\]
Thus Alice's betting strategy can be modelled as a classical channel
\[
\mathcal{B}_A: X \to A', \quad q(\cdot|x) \in \Delta(\mathcal{Z}),
\]
and Bob's odds-setting strategy as
\[
\mathcal{B}_B: Y \to B', \quad r(\cdot|y) \in \Delta(\mathcal{Z}),
\]
where $A'$ and $B'$ are registers holding the respective stake and odds vectors. We can thus think of the act of betting and placing odds, as local channels.
This identification allows us to treat bets and odds as conditional distributions, which is both notationally convenient and prepares the ground for the quantum generalisation to $\rho_{ABC}$.

\section{Sketch of Adversarial Quantum Resource Theories}

Below we sketch relevant details from \cite{ARTArcos2025adversarial} which may be helpful for the reader.

\subsection{A risk-reward trade-off for gambling with side information }

In \cite{arcos2025gambling} a risk-reward trade-off was derived, which quantified how much wealth Alice could gain (the reward), given that she wants to succeed with probability at least $\epsilon$ (the risk). Here, we will extend the result to the case with side informtion, so that we can apply it to the quantum case.
Let us consider another form for the wealth gained, namely
\begin{align}\label{singleshotbayesian}
    \frac{Q^{A}(z^{n}|x^{n})}{Q^{B}(z^{n}|y^{n})}
    =&
     2^{n(D(\lambda_{y^{n}z^{n}}\| \tilde{W}_{YZ} )-D(\lambda_{x^{n}z^{n}}\| \tilde{W}_{XZ})  
    + H_\lambda(Z|Y)-H_\lambda(Z:X)
     )}
\end{align}
where $\tilde{W}_{XZ}(x,z) = \lambda_{x^{n}}(x)Q^{A}(z|x)$, $\tilde{W}_{YZ}(y,z)=\lambda_{y^{n}}(y)Q^{B}(z|y)$ are the distributions defined by Alice's and Bob's strategies respectively, and 
$H_\lambda(Z|Y):=H(\lambda_{y^{n},z^{n}})-H(\lambda_{y^{n}})$, $H_\lambda(Z|X):=H(\lambda_{x^{n},z^{n}})-H(\lambda_{x^{n}})$ can be thought of as a empirical mutual information and is independent of Alice and Bob's strategies.

The gambling rate with side information is
\begin{equation}
   R= D(\lambda_{y^{n}z^{n}}\| \tilde{W}_{YZ} )-D(\lambda_{x^{n}z^{n}}\| \tilde{W}_{XZ})  
    + H_\lambda(Z|Y)-H_\lambda(Z|X)
    \label{eq:wealthrate}
\end{equation}
and we note that the $H_\lambda(Z|Y)-H_\lambda(Z|X)$ term is independent of Alice and Bob's strategy. While it is important for computing the rate Alice's wealth changes, we can drop it when optimising Alice and Bob's strategy. We can then see that we have a risk reward trade-off with the same mathematical structure as hypothesis testing. 
The probability of Alice's bet succeeding $P_{suc}$ is approximately the probability of her chosen type class occurring, which is (up to subexponential factors) 
\begin{align}
  P_{suc}&\approx P(\lambda_{y^nz^n})\approx 2^{-nD(\lambda_{y^nz^n}||P_{YZ})}\geq \epsilon
  \label{eq:risk}
\end{align}
which is the risk (the probability that her bet succeeds). For a fuller discussion of this, we refer the reader to \cite{arcos2025gambling}. If we require that Alice's bet succeeds with a probability greater than $\epsilon$, then this corresponds to a set of typeclasses $\lambda_{y^nz^n}$ which satisfy Eq. \eqref{eq:risk}. If Alice bets according to $\tilde{W}_{XZ}:=\lambda_{x^n}Q^A(z|x)=\lambda_{x^nz^n}$ and if this typeclass occurs, then the rate of increase of $W_F/W_i$ (the reward) is from Eqn \eqref{eq:wealthrate},
\begin{align}
    R=D(\lambda_{y^{n}z^{n}}\| 1_YQ^B )
    + H_\lambda(YZ)-H_\lambda(Z|X)
    \label{eq:reward}
    \end{align}
were $1_Y$ is $1$ for each string $y^n$, such that $\mu_Y:=\frac{1_Y}{|Y|}$ is the uniform distribution over strings $y^n$, and $D(\lambda_{y^{n}z^{n}}\| 1_YQ^B )$ is defined as on probability distributions (i.e. $D(\lambda_{y^{n}z^{n}}\| 1_YQ^B )=D(\lambda_{y^{n}z^{n}}\| \mu_YQ^B )-\log|Y|$)
    
    The optimisation is then as follows: we require that $D(\lambda_{y^nz^n}||P_{YZ})\leq\frac{1}{n}\log\epsilon$ and then want to maximise $D(\lambda_{y^{n}z^{n}}\| 1_YQ^B)$ subject to the risk constraint. I.e. we want to maximise our rate of return when our bet succeeds. This optimisation is precisely that found in asymmetric hypothesis testing: We can think of $P_{YZ}$ as the alternative hypothesis, and $\mu_YQ^B$ the null hypothesis. Given a bound on the type I error (false positives), we want to minimise the type II error (false negatives). This optimisation has a well known solution\cite{neyman1933efficient,CoverThomas2006}.

\section{Quantum gambling} 

We have reviewed gambling in a classical settings, where Alice and Bob interact in games with partial information about the outcome of a race $Z$. 
We reviewed how Alice's wealth evolves over time. In \cite{arcos2025gambling}, gambling over a finite number of rounds was considered, and these results were connected with expected utility theory, and hypothesis testing. These classical frameworks provide valuable insights into decision-making under uncertainty, but they do not account for the unique features of quantum systems, such as entanglement and superposition, which are crucial for understanding quantum resource theories. 

In this section, we begin to extend the adversarial gambling framework to the quantum domain, where Alice, Bob, and Charlie share quantum states, and their strategies depend on quantum correlations shared between their systems. Since considering quantum states shared between three parties is complex, we first begin by considering the notion of a quantum adversarial game when the adversary is passive. We will find that this corresponds to quantum state merging with classical communication. We will then consider active adversaries, and find this corresponds to adversarial quantum state merging using variable length quantum codes, in a way which is similar to the classical gambling game using variable length classical codes described in Section \ref{sec:GamblingRT}.

\subsection{Quantum gambling with a passive adversary: state-merging}

In quantum gambling we replace $P_{XYZ}$ the classical side information used in classical gambling, with $n$ copies of a pure state $|\psi\rangle_{ABC}$. Classical communication is free. Here, setting the odds is implicit, since Bob's state forces Charlie to use more entanglement. Here we have a risk reward trade-off. The probability of Charlie to succeed in state-merging his share $C$ to Alice, is required to be greater than some $\epsilon$, and then we want to maximise (minimise), the amount of pure state entanglement Alice and Charlie 
obtain (require).

In the one-shot case, the entanglement consumed to carry out state merging~\cite{HOW2005} for a state $\rho_{AC}$ is tightly characterised by the smooth max-entropy. Specifically, the number~$N$ of ebits required (or obtained, if $N < 0$) to achieve state merging with success probability $1-\varepsilon$ satisfies~\cite{DBWR2014}
\begin{align} \label{eq:statemerging}
H_{\max}^{\varepsilon'}(C|A)_\rho - 2 \log(1/\varepsilon) \leq N \leq H_{\max}^{\varepsilon''}(C|A) + 4 \log(1/\varepsilon)_\rho + \mathrm{const} ,
\end{align}
with smoothing parameters $\varepsilon'$ and $\varepsilon''$ close to $\varepsilon$ (specifically, $\log(1/\varepsilon') = \mathrm{const'} \log(1/\varepsilon)$ and $\log(1/\varepsilon'') = \mathrm{const''} \log(1/\varepsilon)$). 

The smooth max-entropy can be related to hypothesis testing. We have
\begin{align}
    H_{max}^\epsilon(C|A)_\rho&=-H^\epsilon_{min}(C|B)\label{eq:becPure}\\
    &:=\inf_{\sigma_B}D^\epsilon_{max}(\rho_{BC}||\mathrm{id}_C\otimes\sigma_B)\label{eq:becDef}\\
   & \approx -H_H^\epsilon(C|B)\label{eq:becRLD}
\end{align}
where Eq \eqref{eq:becPure} follows from the purity of $|\psi\rangle_{ABC}$,and Eqn \eqref{eq:becRLD} follows from ~\cite{RLD2025},
where the approximation is to be understood up to additive terms of the order $\log(1/\varepsilon)$ as in~\eqref{eq:statemerging}, and where
\begin{equation}
    H_H^\varepsilon(C|B)_{\rho} = -\inf_{\sigma_B} D_H(\rho_{B C} \| \mathrm{id}_C \otimes \sigma_B) \ .    
    \label{eq:hyptestentropy}
\end{equation}
The quantity on the right hand side $D_H$, is the \emph{hypothesis testing entropy}, defines as the error exponent in hypothesis testing. Concretely, for any states $\rho$ and $\sigma$,
\begin{equation}
    2^{-D_H^\varepsilon(\rho \| \sigma)} := \frac{1}{\varepsilon} \inf \{ \tr(\Pi\sigma): 0 \leq \Pi \leq \mathrm{id} \wedge \tr(\Pi \rho) \leq \varepsilon\} \ .
\end{equation}
Combining this with the above relations thus shows that the amount of entanglement required (obtained) by state merging is directly related to hypothesis testing, and of a form which is similar to the risk reward trade-off in the case of classical gambling with side-information, given by Eqns. \eqref{eq:risk} and \eqref{eq:reward}. In the classical case, $P_{YZ}$ was the alternative hypothesis, and $\mu_YQ^B$ the null hypothesis. Here, $\rho_{BC}$ is the alternative hypothesis and $\frac{\mathrm{id}_C}{d_C}\otimes\sigma_B$. Alice fixes a probability by which she wants to win the bet (obtain $C$ and some number of ebits), and given that restriction she optimises for the number of ebits she wins when she succeeds. Alternatively, she can fix the number of ebits she wants to win, and then find the protocol which has the highest probability of success for that reward.

The expressions for classical and state-merging risk-reward trade-off thus look almost identical, but there is a crucial difference. In the classical case, Bob's odds $Q^B$ are set by the adversary and are arbitrary. Alice can place a long shot bet, and win an arbitrarily large amount of money, by betting on a type class $\lambda_{y^nz^n}$ such that $D(\lambda_{y^nz^n}||1_YQ^B)$ is very large. On the other hand, in state-merging, we take the infimum over all $\sigma_B$, and the number of ebits that Alice wins will not be superinear in $n$. 

\subsection{Quantum gambling as variable length state-merging}

Let us recall that in looking for a resource theory of classical gambling in \cite{arcos2025gambling}, it turned out that we had to allow for communication from Charlie to Alice, and this was the valuable resource which Alice and Bob used as currency in their gambling game. This was required because correlation itself was not enough -- it couldn't account for long shot bets where Alice wins a superlinear amount of money, and arbitrary odds by Bob. The analogy carries through to the quantum case. It is not enough to consider Local Operations and Classical Communication, and entanglement manipulation, since this gives us the restricted setting found with state-merging  -- there is no active adversary who can set arbitrary odds, and Alice cannot place long shot odds.

The solution, analogous to the classical case, is to allow for both parties to bet communication resources, as with the variable length coding scheme discussed in Section \ref{sec:GamblingRT}. We again consider many copies of $|\psi\rangle_{ABC}$, but this time, we imagine Alice, Bob and Charlie each performing a measurement on their local state. In order to keep the game fully quantum, he will perform the measurement coherently, so that at the end of the protocol, the state $|\psi\rangle_{ABC}^{\otimes n}$ will remain pure, with Charlie's share transfered to Alice. This requires the parties to perform operations which are contingent on the outcome of the measurement, but they must all be done coherently. The operations that correspond to gambling consiste of the following:
\begin{itemize}
    \item Bob will perform a coherent and incomplete measurement on the eigenbasis $|y^n\rangle_B$ of his local quantum state, to obtain the typeclass of $y^n$, $\lambda_{y^n}$. We may allow more general coherent measurements, but we will not consider this here.   Conditional on $\lambda_{y^n}$ he will then set the odds for each possible outcome of Charlie's coherent measurement, which we will label as $|z^n\rangle_C$. 
    Although we write $\ell_B(z^n|y^n)$ as dependent on the strings themselves, they will have the same value for strings $z^n$ and $y^n$ which share the same typeclass $\lambda_{y^n}$, $\lambda_{z^n}$, since strings with the same typeclass are all equally likely to occur.
    Analogously to the classical gambling resource theory, the odds correspond to a promise of $\ell_B(z^n|x^n)$ ebits between Alice and Charlie, for each outcome $|z^n\rangle_C$. In order that Alice receive the state unambiguously, these lengths will need to correspond to a prefix free code and must satisfy the Kraft-McMillian equality
\begin{align}
    \sum_{z^n}2^{-\ell_B(z^n|y^n)}=1
\end{align}
    \item Alice performs a similar measurement onto the typeclass $\lambda_{x^n}$ in the eigenbasis of her state $|x^n\rangle_A$. Conditioned on $\lambda_{x^n}$, she places a bet, corresponding to how many ebits $\ell_A(z^n|x^n)$ she will need to to receive the state from Charlie, contingent on the outcome of his coherent measurement. Once again, these only depend on the typeclasses of $x^n$ and $y^n$, and must satisfy
    \begin{align}
    \sum_{z^n}2^{-\ell_A(z^n|x^n)}=1
    \end{align}
in order that Alice unambiguously receives Charlie's state.
    \item Charlie measures his typeclass $\lambda_{z^n}$ in his eigenbasis $|z^n\rangle_C$.
    Conditioned on all these outcome, Alice, Bob and Charlie share some state $|\psi^{\lambda}\rangle_{ABC}$, where $\lambda$ labels the joint typeclass $\lambda_{x^ny^nz^n}$.
    \item We now wish to coherently add or subtract, a number of additional ebits to our state, conditional on $\lambda_{x^ny^nz^n}$. To do this, we can use an embezzling state shared by Alice and Charlie, in a manner very similar to \cite{berta2011quantum}, who considered a related state-splitting protocol in the context of the quantum reverse Shannon Theorem. In particular, we can consider the embezzling state to be tantamount to an ''entanglement battery''\cite{alhambra2019entanglement}. If Bob assigns odds $\ell_B{z^n|y^n}$, to each outcome $z^n$, then this corresponds to him adding  $\ell_B{z^n|y^n}$ to Alice and Charlie's entanglement battery, conditional on $\lambda_{y^nz^n}$. On the other hand, Alice may only require $\ell_A{z^n|x^n}$ ebits for Charlie to state-merge  the $C$ share of $|\psi^{\lambda}\rangle_{ABC}$, to her. 
\end{itemize}

 Alice's cost for each $|\psi^\lambda\rangle_{ABC}$ branch of the protocol, $\ell_A(\lambda)$, is the number of ebits required to perform one-shot state merging on the post-measurement state $|\psi^\lambda\rangle_{ABC}$. This cost is given by the smooth conditional max-entropy, $\ell_A(\lambda) \approx H_{max}^{\varepsilon}(C|A)_{\psi_\lambda}$. The difference between the number of ebits she needs, and the number of ebits Bob has promised, $k_{win}(\lambda) = \ell_B(\lambda) - \ell_A(\lambda)$, is her winnings or loss. This corresponds to the net number of ebits deposited into (or withdrawn from) Alice's entanglement battery. 
 In the asymptotic limit, the number of ebits in Alice's entanglement batter  will increase at a rate of $H(C|A)-H(C|B)$, which correspond to the asymptotic Kelly betting rate, in the case of side information.

\printbibliography

\end{document}